\documentclass[conference]{IEEEtran}
\IEEEoverridecommandlockouts
\usepackage{cite}
\usepackage{amsmath,amssymb,amsfonts}

\usepackage{amsthm}  

\newtheorem{proposition}{Proposition}
\theoremstyle{remark}

\usepackage{float}
\usepackage{graphicx}
\usepackage{textcomp}
\usepackage{xcolor}
\usepackage{balance} 
\usepackage{algorithm}
\usepackage{algorithmicx}
\usepackage{algpseudocode}
\usepackage{float}
\makeatletter
\renewcommand{\fps@table}{H}
\makeatother
\usepackage[section]{placeins}
\usepackage{placeins}
\usepackage{enumitem}
\usepackage{graphicx}
\usepackage{caption}
\usepackage{subcaption}
\usepackage{booktabs}   
\usepackage{multirow}   
\usepackage{amsmath}    
\usepackage{amssymb}    
\usepackage[utf8]{inputenc}
\def\BibTeX{{\rm B\kern-.05em{\sc i\kern-.025em b}\kern-.08em
    T\kern-.1667em\lower.7ex\hbox{E}\kern-.125emX}}
\begin{document}

\title{DISPATCH - Decentralized Informed Spatial Planning and Assignment of Tasks for Cooperative Heterogeneous Agents

}

\author{
\IEEEauthorblockN{
Yao Liu\IEEEauthorrefmark{1},
Sampad Mohanty\IEEEauthorrefmark{2},
Elizabeth Ondula\IEEEauthorrefmark{3}, and
Bhaskar Krishnamachari\IEEEauthorrefmark{1}\IEEEauthorrefmark{2}
}
\IEEEauthorblockA{
\IEEEauthorrefmark{1}Department of Electrical Engineering, University of Southern California, Los Angeles, CA, USA\\
\IEEEauthorrefmark{2}Department of Computer Science, University of Southern California, Los Angeles, CA, USA\\
\IEEEauthorrefmark{3}Department of Computer Science, Winthrop University, Rock Hill, SC, USA\\
Email: \{jliu3548, smohanty, bkrishna\}@usc.edu, ondulae@winthrop.edu
}
}
\makeatletter
\def\@maketitle{%
  \newpage
  \null
  \vskip 0.5em
  \begin{center}%
  {\LARGE\bfseries \@title\par}
  \vskip 0.5em
  {\lineskip .5em
    \begin{tabular}[t]{c}%
      \@author
    \end{tabular}\par}%
  \vskip 1em%
  \end{center}%
  \par
}
\makeatother

\IEEEpubid{\makebox[\columnwidth]{\hfill\begin{minipage}{\columnwidth}
\centering\scriptsize
This work has been submitted to the IEEE for possible publication.\\
Copyright may be transferred without notice, after which this version may no longer be accessible.
\end{minipage}}\hspace{\columnsep}\makebox[\columnwidth]{}}

\IEEEpubidadjcol

\maketitle
\IEEEpubidadjcol

\begin{abstract}
\IEEEpubidadjcol
Spatial task allocation in systems such as multi-robot delivery or ride-sharing requires balancing efficiency with fair service across tasks. Greedy assignment policies that match each agent to its highest-preference or lowest-cost task can maximize efficiency but often create inequities: some tasks receive disproportionately favorable service (e.g., shorter delays or better matches), while others face long waits or poor allocations.

We study fairness in heterogeneous multi-agent systems where tasks vary in preference alignment and urgency. Most existing approaches either assume centralized coordination or largely ignore fairness under partial observability. Distinct from this prior work, we establish a connection between the Eisenberg–Gale (EG) equilibrium convex program and decentralized, partially observable multi-agent learning. Building on this connection, we develop two equilibrium-informed algorithms that integrate fairness and efficiency: (i) a multi-agent reinforcement learning (MARL) framework, EG-MARL, whose training is guided by centralized EG equilibrium assignment algorithm; and (ii) a stochastic online optimization mechanism that performs guided exploration and subset-based fair assignment as tasks are discovered.

We evaluate on Multi-Agent Particle Environment (MPE) simulations across varying team sizes against centralized EG, Hungarian, and Min–Max distance baselines, and also present a Webots-based warehouse proof-of-concept with heterogeneous robots. Both methods preserve the fairness–efficiency balance of the EG solution under partial observability, with EG-MARL achieving near-centralized coordination and reduced travel distances, and the online mechanism enabling real-time allocation with competitive fairness. Together, these results demonstrate that spatially aware EG formulations can effectively guide decentralized coordination in agents with heterogeneous capabilities.

\end{abstract}


\section{Introduction}

Teams of autonomous agents are increasingly deployed to tackle missions where tasks are spread across space and subject to strict resource constraints. Such domains include disaster response (e.g., search and rescue)~\cite{Bhute2024ImplementingSR}, swarm environmental monitoring~\cite{survey_coop2025,neurofleets2024}, and planetary exploration~\cite{planetary2023}. In these applications, agents must navigate through complex environments and perform service operations at designated sites. Tasks often require \textit{diverse capabilities or skill levels}—for instance, transport, delicate manipulation, or specialized inspection—making assignment both a problem of navigation and matching heterogeneous agents to compatible service demands. These challenges are further compounded by travel and service costs, limited energy resources, and environmental uncertainty~\cite{zhang2025distributed}.

A key challenge in multi-agent task allocation is balancing efficiency with fairness across tasks. Prior approaches in robotics and multi-agent planning aim to maximize collective performance~\cite{framework_simultaneous2023}. However, methods that solely optimize overall system efficiency often create inequities: some tasks are completed quickly and matched to highly capable agents, while others face long delays or poor matches. Chin et al. \cite{chin2024drone} studied drone-base placement for post-disaster search and rescue and show that purely efficiency-driven deployment can create substantial inequities in access across affected regions. In mission-critical or human-facing domains, such disparities can undermine reliability and erode trust in autonomous systems. Imposing fairness constraints can reduce system efficiency, highlighting the fundamental fairness–efficiency trade-off in multi-agent decision-making~\cite{ranjan2025fairness}. This is especially pronounced in spatial, dynamic, and online settings, where tasks differ in urgency, workload, and compatibility with agents.

We study a multi-agent navigation-and-service problem in a two-dimensional workspace with static obstacles, where $n$ heterogeneous agents must reach and serve $m$ spatially distributed tasks under partial observability. Tasks have workload demands and importance weights, while agents differ in service capabilities, captured through a preference vector encoding agent--task compatibility. The goal is to minimize total completion time and travel cost while ensuring fair service across tasks given their preferences and importance.
\begin{figure}[t]
    \centering
\includegraphics[width=0.75\linewidth]{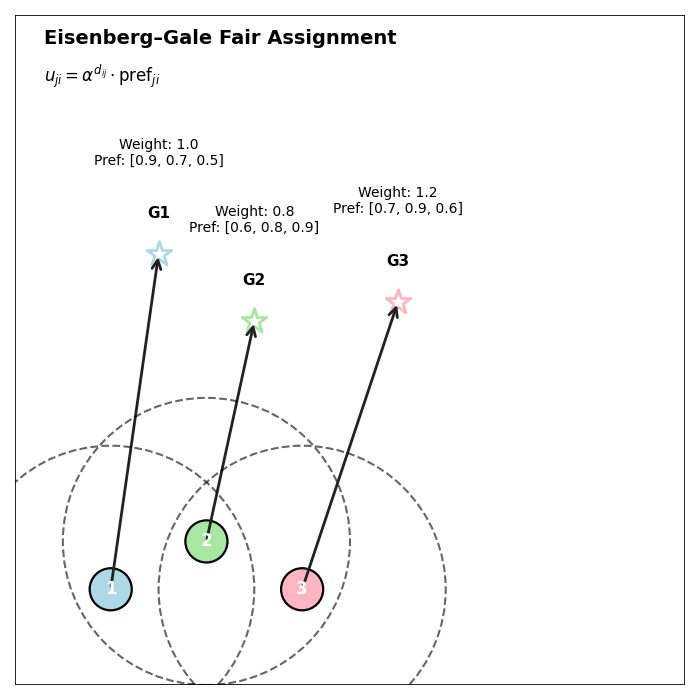}
    \caption{Illustration of the heterogeneous multi-robot navigation-and-service task allocation problem.}
    \label{fig:illustrative_intro}
\end{figure}

To address these challenges, we connect equilibrium theory with decentralized multi-agent coordination. We extend the Eisenberg--Gale framework to embodied agents with spatial costs and partial observability, showing how its fairness structure can guide decentralized policy learning. This yields two complementary algorithms: (i) a CTDE MARL framework that uses a spatially modified EG program as a centralized teacher, and (ii) an EG-inspired online assignment method that provides a centralized baseline. This differs from prior market-based methods such as FMC-TA~\cite{amador2014dynamic}, FMC\_TA+~\cite{cooperative_task_alloc}, and online Fisher/EG pricing~\cite{jalota2024stochastic}, which assume centralized global information, as well as from fair-MARL approaches~\cite{zhang2014fairness,jiang2019learningfairness,aloor2024cooperation}, which rely on reward shaping or max--min objectives without an equilibrium-based fairness--efficiency reference.

\textbf{Fairness Notions.}
We adopt two complementary notions of fairness. First, the centralized Eisenberg--Gale (EG) program provides a \emph{Pareto-efficient} benchmark. Because EG-MARL operates under partial observability, we measure deviation from this benchmark using \emph{regret}, defined as the gap between its realized EG objective and the centralized optimum. 

Second, we evaluate \emph{distributional fairness} via the normalized utility-to-weight ratio $\rho_j = u_j^{\text{realized}} / w_j$, with uniformity quantified by $F(\rho)=1/\mathrm{CV}(\rho)$. We additionally report Jain's Fairness Index,
\[
J(\rho)=\frac{(\sum_j \rho_j)^2}{m\sum_j \rho_j^2},
\]
which equals 1 under perfectly equal service and decreases as disparities increase.
Our setting requires both spatially aware navigation and skill-aligned assignment, where skill compatibility is modeled through the preference coefficients. Our main contributions are summarized as follows:
\begin{itemize}[leftmargin=*,topsep=2pt,itemsep=1pt,parsep=0pt]
    \item \textbf{Spatial-aware Eisenberg-Gale:} We extend the Eisenberg--Gale (EG) framework by incorporating spatial travel costs into the utility formulation, bridging classical equilibrium theory with multi-agent task allocation.
    \item \textbf{EG\mbox{-}MARL.} We establish a centralized training decentralized execution (CTDE) framework that integrates graph-based policies with a spatially aware Eisenberg–Gale (EG) assignment, allowing agents to learn decentralized policies that jointly optimize fairness and travel efficiency under partial observability.
    \item \textbf{Stochastic Online Assignment.} We propose a cooperative explore–and–assign mechanism that operates under centralized communication. Whenever k new tasks are discovered, an Eisenberg–Gale (EG) instance is solved over the currently unassigned agents and discovered tasks. Assigned agents immediately begin servicing their targets, while the remaining agents continue exploring, and this process repeats until all tasks are completed.
    \end{itemize}

We validate our approach in Multi-Agent Particle Environment (MPE) simulations across multiple team sizes and in a Webots warehouse (See Appendix), proof-of-concept with heterogeneous robots.

\section{Background and Related Work}
In multi-agent settings, fairness was studied under planner-driven optimization, where a single controller allocates resources to maximize social welfare~\cite{chen2021welfare}.  Lan and Chiang~\cite{lan2010axiomatic} established an axiomatic foundation for fairness through $\alpha$-fairness and proportional fairness, framing it as a welfare-maximization problem. Market-based formulations such as the Fisher market~\cite{fisher1892mathematical} and the Eisenberg--Gale (EG) program~\cite{eisenberg1959eg} later introduced equilibrium mechanisms that guarantee Pareto efficiency and envy-freeness under budget constraints~\cite{budish2011coursealloc,chaudhury2024chores}. More recently, Kumar and Yeoh proposed DECAF~\cite{kumar2025decaf}, which uses DDQN to learn centralized allocation policies that trade off total utility and long-term fairness under multiple fairness objectives. To address temporal and spatial heterogeneity, Amador et al.~\cite{amador2014dynamic} proposed a Fisher-market-inspired algorithm (FMC-TA) that produces dynamic Pareto-efficient assignments incorporating travel costs. Building on this foundation, Amador and Zivan~\cite{cooperative_task_alloc} introduced \textit{FMC\_TA+}, which extends the market-clearing formulation to concave personal utilities and heterogeneous agent skills. Their results showed that slight concavity encourages cooperative task execution among diverse agents while maintaining Pareto-optimal and envy-free equilibria. Zhang and Shah~\cite{zhang2014fairness} introduced a regularized max--min fairness criterion for factored multi-agent MDPs, linking fairness and efficiency, though their formulation assumes centralized control and full observability.However, the aforementioned frameworks assumed complete global information. More recently, Jalota and Ye~\cite{jalota2024stochastic} generalized EG and Fisher markets to stochastic and online environments through adaptive pricing updates, ensuring near-optimal efficiency and bounded regret. Zimmer et al.~\cite{zimmer2021} propose a decentralized cooperative MARL framework that balances efficiency and equity by optimizing a social welfare function, enabling inequality-sensitive objectives such as max--min and generalized Gini--based fairness. These advances mark a shift from static, centralized equilibria toward dynamic and sequential decision-making, where fairness must emerge through local interactions and learning rather than global coordination.

Complementary to these equilibrium-inspired methods, many works introduce fairness through \textit{reward shaping}, adjusting agents' rewards to penalize disparities or promote balanced contributions~\cite{reuel2024fairness}. Jiang and Lu~\cite{jiang2019learningfairness} proposed a decentralized hierarchical Multi-Agent Reinforcement Learning method that encourages equitable outcomes by optimizing a shaped “fair-efficient” reward combining system utility with a penalty for agents’ deviation from the population-average utility. Aloor et al.~\cite{aloor2024cooperation} further showed that distance-based shaping can enforce equitable service distribution in spatial coordination tasks.

Our work builds upon and extends the Fair-MARL framework of Aloor et al.~\cite{aloor2024cooperation}, retaining their distance-based fairness reward while generalizing it to heterogeneous, skill-aware agents. We further incorporate graph-based representations inspired by InforMARL~\cite{nayak2022scalable}, enabling scalable decentralized execution under partial observability. By embedding the economic fairness guarantees of Eisenberg--Gale markets within this decentralized MARL structure, our framework achieves equitable and efficiency-aware task allocation in spatial, dynamic environments with limited communication.

\section{Preliminaries}

\paragraph{\textbf{Navigation-Service Task}}
\label{sec:obs}
We model the multi-agent\\  navigation-and-service task as a Decentralized Partially Observable Markov Decision Process (Dec-POMDP), where $N$ heterogeneous agents and $N$ tasks are distributed in a two-dimensional workspace with static walls and movable obstacles. Each agent operates under partial observability, aiming to reach and service a single task, with assignments guided by task weights, preferences, and travel distances.

Tasks are static in location but are discovered online as agents explore the environment; an undiscovered task becomes visible only when it enters an agent’s sensing radius. A task is considered completed when its workload has been fully reduced to zero under the preference-weighted service dynamics.

Let $\mathcal{A}$ denote the set of agents and $\mathcal{T}$ the set of tasks (or tasks), with $|\mathcal{A}| = |\mathcal{T}| = N$. Tasks are characterized by workload and priority, while agents differ in service capabilities. The preference parameter $p_{ji}$ quantifies how effectively agent $i$ serves task $j$ and thus governs assignment compatibility. It directly controls how quickly the workload of task $j$ decreases when served by agent $i$, linking compatibility to completion rate. For clarity, we denote this efficiency as $\text{pref}_{ji}$ in workload and progress modeling.

\paragraph{\textbf{Definition of Dec-POMDP}} 
The environment is modeled as a Decentralized Partially Observable Markov Decision Process (Dec-POMDP) defined by the tuple:
$
\langle N,\, \mathcal{S},\, \mathcal{O},\, \mathcal{A},\, \mathcal{G},\, P,\, R,\, \gamma \rangle,
$
where: 
\begin{itemize}[leftmargin=*,topsep=2pt,itemsep=1pt,parsep=0pt]
    \item $N$ is the number of agents.
    \item $\mathcal{S}$ is the state space, including all agents’ and landmarks’ positions, velocities, types, workloads, occupancy status, and obstacle locations.
    \item $\mathcal{O}$ is the local observation space; each agent $i$ receives an egocentric observation $o^{(i)}$ of entities within its sensing radius $r_i$.
    \item $\mathcal{A}$ is the discrete action space; each agent selects unit acceleration/deceleration along $x$ and $y$, subject to velocity limits (including an idle action).
    \item $\mathcal{G}$ builds, for each agent $i$, a local graph $g^{(i)}=\mathcal{G}(s;i)$ of nearby entities. Nodes are entities within $i$’s sensing radius $r_i$; edges link nodes that lie within $r_i$.

    \item $P(s' \mid s, A)$ is the transition kernel that updates positions, velocities, workloads, and task states given the joint action $A=(a^{(1)},\dots,a^{(N)})$, subject to obstacles and walls.
    \item $R(s, A)$ is the joint reward, combining workload progress, navigation-distance reduction, collision penalties, and fairness shaping from the Eisenberg--Gale allocation.
    \item $\gamma \in (0,1]$ is the discount factor.
\end{itemize}

We follow the joint-reward formulation of the Multi-Agent Particle Environment~\cite{lowe2017multi}, where the global reward at each timestep is the sum of all agents’ individual rewards:
\[
R(s_t, A_t)=\sum_{i=1}^{N} r_t^{(i)}.
\]
which encourages cooperative behavior among agents. 

 \paragraph{Agent–Entity Graph.} We adopt the local graph construction from InforMARL~\cite{nayak2022scalable}, which models agent--entity interactions through directed and bidirectional edges to represent perception and communication relationships. At each timestep, every agent $i$ forms a local agent–entity graph $g^{(i)} = (\mathcal{V}^{(i)}, \mathcal{E}^{(i)})$ consisting of all entities (agents, tasks, obstacles) within its sensing radius. Each node $v \in \mathcal{V}$ is an entity in the environment, and each entity is associated with a type:
 \[
 \text{entity\_type}(j) \in \{\text{agent}, \text{obstacle}, \text{task}\}.
 \]
 An edge $e \in \mathcal{E}$ is created between an agent and an entity if the entity lies within the agent’s sensing radius $\rho$. Agent–agent edges are bidirectional, while agent–non-agent edges are directed from the entity toward the agent, reflecting perception. Thus, a unidirectional edge encodes one-way perception from an agent to a nearby entity, while a bidirectional edge captures two-way communication and information sharing between agents. The GNN performs message passing over this graph, producing a relational embedding that captures the structure and attributes of the agent's neighborhood.


\paragraph{\textbf{Centralized Training with Decentralized Execution (CTDE)}}
CTDE enables agents to train using global information during training but operate with local observations at execution time, and has been widely discussed in recent surveys~\cite{amato2024ctde, wong2023deepmultiagent}.
Our CTDE formulation builds on the design principles introduced in InforMARL~\cite{nayak2022scalable}, which demonstrated that local graph-based message passing can enable decentralized policies to scale under partial observability.

During \emph{training}, a centralized critic has access to the global state and all agents’ local graphs, providing value estimates and policy gradients based on full system information. Each agent’s policy $\pi^{(i)}_\theta(a^{(i)}|o^{(i)}, g^{(i)})$ is optimized to maximize the expected discounted reward:
\[
J(\theta) = \mathbb{E}_{A_t, s_t} \Bigg[ \sum_t \gamma^t R(s_t, A_t) \Bigg].
\]
During inference, agents execute their learned policies independently, using only their local observations, communications, and corresponding agent–entity graphs.

\paragraph{\textbf{Eisenberg-Gale Optimization}}

Let $w_j>0$ be the importance weight associated with task $j$, reflecting its priority or urgency.  
We define the effective utility for agent $i$ serving task $j$ by combining preferences with  distance discounting:
\[
u_{ji} = (\alpha^{d_{ij}})\,\text{preference}_{ji}
\]
where $d_{ij}$ is the Euclidean distance between agent $i$ and task $j$, $\alpha \in (0,1)$ is a discount factor penalizing longer travel, and $\text{preference}_{ji}$ reflects how well agent $i$'s skills align with task $j$’s requirements.

We solve the following convex program~\cite{eisenberg1959eg}:
\begin{align}
    \max_{x \ge 0} \quad & \sum_{j \in \mathcal{T}} w_j \log\!\Bigg(\sum_{i \in \mathcal{A}} u_{ji}\,x_{ji}\Bigg)
    \label{eq:eg-objective} \tag{1a} \\
    \text{s.t.}\quad
    & \sum_{j \in \mathcal{T}} x_{ji} = 1, \quad \forall i \in \mathcal{A}
    \label{eq:eg-constraint1} \tag{1b} \\
    & \sum_{i \in \mathcal{A}} x_{ji} = 1, \quad \forall j \in \mathcal{T}
    \label{eq:eg-constraint2} \tag{1c}
\end{align}

\paragraph{Remark}
Although the Eisenberg--Gale (EG) objective in~\eqref{eq:eg-objective} appears nonlinear due to the logarithmic term, it admits an equivalent \emph{linear assignment form} when restricted to one-to-one allocation constraints~(\ref{eq:eg-constraint1})--(\ref{eq:eg-constraint2}). 
Since each feasible matrix $x$ corresponds to a permutation matrix in which every row and column contains exactly one nonzero entry, the inner sum $\sum_i u_{ji}x_{ji}$ for each task $j$ simply selects a single utility value $u_{j i^*}$ corresponding to the assigned agent $i^*$. 
Consequently, the logarithmic term $w_j \log\!\big(\sum_i u_{ji}x_{ji})$ becomes a constant coefficient $w_j \log(u_{j i^*})$ for that assignment, yielding the equivalent linear objective
\[
\max_{x \in \{0,1\}^{N \times N}} \sum_{i,j} s_{ji}\, x_{ji}, 
\quad \text{where } s_{ji} = w_j \log(u_{ji} + \varepsilon).
\]
Thus, under one-to-one matching constraints, the EG program is \emph{linear in disguise}: it can be solved as a linear program with no integrality gap, or equivalently through the Hungarian algorithm with the precomputed score matrix $S = [s_{ji}]$. 

\begin{proposition}[Pareto Efficiency of the One-to-One Eisenberg--Gale Allocation]
The solution $x^*$ to the one-to-one Eisenberg--Gale program is Pareto efficient among all feasible (binary) assignments.
\end{proposition}
\begin{proof}
Let $v_j(x) = \sum_i u_{ji} x_{ji}$ denote the realized utility of task $j$ under allocation $x$. 
Because the logarithmic function is strictly increasing, any alternative feasible one-to-one allocation $x'$ 
that improves at least one task’s utility without reducing others’ would satisfy
\[
\sum_j w_j \log(v_j(x')) > \sum_j w_j \log(v_j(x^*)),
\]
contradicting the optimality of $x^*$. 
Hence, $x^*$ is Pareto efficient \emph{within the discrete feasible set} of one-to-one assignments.
\end{proof}

We provide two baselines for the EG formulation. 
The first is the Hungarian method~\cite{kuhn1955hungarian}, 
which seeks the utilitarian optimum by maximizing total linear preference between agents and tasks:
\begin{align}
    \max_{x \in \{0,1\}^{N \times N}} \quad & \sum_{i \in \mathcal{A}}\sum_{j \in \mathcal{T}} {pref}_{ji}\, x_{ji} \label{eq:hungarian-objective} \tag{2a}\\
    \text{s.t.}\quad 
    & \sum_{j \in \mathcal{T}} x_{ji} = 1, \quad \forall i \in \mathcal{A} \label{eq:hungarian-constraint-rows} \tag{2b}\\
    & \sum_{i \in \mathcal{A}} x_{ji} = 1, \quad \forall j \in \mathcal{T} \label{eq:hungarian-constraint-cols} \tag{2c}
\end{align}

We also include the min-max baseline\cite{aloor2024cooperation}, 
which promotes equity by minimizing the maximum individual distance cost. 
Here, $c_{ji}$ denotes the distance cost (e.g., Euclidean distance) for agent $i$ to reach task $j$.
\begin{align}
    \min_{x \in \{0,1\}^{N \times N}} \quad & \max_{i \in \mathcal{A},\, j \in \mathcal{T}}\, c_{ji} x_{ji} \label{eq:minmax-objective} \tag{3a}\\
    \text{s.t.} \quad 
    & \sum_{j \in \mathcal{T}} x_{ji} = 1, \quad \forall i \in \mathcal{A} \label{eq:minmax-constraint-rows} \tag{3b}\\
    & \sum_{i \in \mathcal{A}} x_{ji} = 1, \quad \forall j \in \mathcal{T} \label{eq:minmax-constraint-cols} \tag{3c}
\end{align}
This baseline minimizes the worst-case distance traveled among agents, ensuring equitable task distribution but lacking preference awareness.


\section{EG-MARL Framework}

\paragraph{\textbf{State, Observations}}

In a centralized Multi-Agent EG framework, agents have full access to the global state, including all agent positions, task locations, workloads, importance weights, and the global preference matrix linking agents to tasks. While this simplifies model design, the assumption of global visibility is unrealistic and limits scalability in larger or more distributed systems. By contrast, our environment avoids these limitations: individual agents do not observe the entire state but instead receive ego-centric observation vectors $o^{(i)}$ designed to enable decentralized, task-aware decision making under partial observability. 

Each agent $i$ observes its state with the geometric and semantic attributes of nearby tasks within its sensing range (feature vectors zero-padded for tasks that are not observed).

A task becomes observable only via local sensing or communication.

Let $\mathbf{p}_i \in \mathbb{R}^2$ and $\mathbf{v}_i \in \mathbb{R}^2$ denote the agent’s position and velocity in the global coordinate frame, and let $\text{type}_i$ and $\text{type}_j$ represent the categorical indices of agent $i$ and task $j$, respectively. 
For every task $j$, the relative position between the task and agent is given by $\mathbf{p}^j_i = \mathbf{p}^{\text{task}}_j - \mathbf{p}_i$, and the global preference matrix encodes the preference profile of all tasks all agents and the preference or efficiency $\text{preference}_{ji}$ of agent $i$ toward task $j$. 
The column vector $\mathbf{p}^\text{pref}_j\ $ corresponds to task $j$ 's preference profile of all agent types, while $w_j$ denotes the importance weight associated with the task type. 
To reflect the dynamic interaction between agents and tasks, we adapt a continuous occupancy variable $\eta_j \in [0,1]$ from Aloor et~al.~\cite{aloor2024cooperation}, defined as $\eta_j = 1 - \min_{i \in \mathcal{N}} \| \mathbf{p}_i - \mathbf{p}_j^{\text{task}} \|_2$, which increases as any agent approaches task $j$. 
A task history term $h_j$ records the last agent to interact with task $j$. 
For a single task $j$, the observation vector of agent $i$ is defined as:
\[
o^{(i)}_j = [p_i^j,\; u_{ji},\; p^{\text{pref}}_j,\; \eta_j,\; h_j,\; w_j].
\]

Following the observation space, we introduce how the agent’s task interactions translate into workload state dynamics. The agent’s workload state is represented by 
$\text{ego\_workload}_i = [w_j,\, w_j^{\text{progress}}(t),\, w_j^{\text{remain}}(t)]$, 
where $w_j$ is the total workload initially assigned to task $j$, 
$w_j^{\text{progress}}(t)$ tracks cumulative service completed by agent $i$, 
and $w_j^{\text{remain}}(t)$ is the remaining workload. 
When agent $i$ services task $j$, the \emph{remaining workload} evolves as
\[
w_j^{\text{remain}}(t + \Delta t) = w_j^{\text{remain}}(t) - \text{preference}_{ji}\,\Delta t,
\]
where $w_j^{\text{remain}}(t)$ denotes the remaining workload of task $j$ at time $t$, and $\text{preference}_{ji}$ is the efficiency of agent $i$ for task $j$.

\paragraph{\textbf{Graph Embeddings}}
\label{sec:obs}

Within each agent’s sensing radius, the environment may include other agents, tasks, obstacles, or walls. The ego agent $i$ constructs a local agent–entity graph $g^{(i)}$ that contains all entities within this region. For each entity $j$, a node embedding $x^{(i)}_j$ is formed in the agent’s local coordinate frame as
\[
x^{(i)}_j =
[\, 
\mathbf{v}_i^j,\;
\mathbf{p}_i^j,\;
\mathbf{p}_i^{\text{goal},j},\;
a_j,\;
t_j,\;
\text{pref}_{ji}^{\text{task}},\;
w_j,\;
\mathbf{c}_j,\;
e_j
\,],
\]
where $\mathbf{v}_i^j$ and $\mathbf{p}_i^j$ represent the relative velocity and position of entity $j$ with respect to agent $i$, and $\mathbf{p}_i^{\text{goal},j}$ denotes the relative location of the goal entity. The categorical variables $a_j$, $t_j$, and $e_j$ correspond to the agent type, landmark type, and entity class, respectively. All entities share this uniform embedding structure, though each type populates only the fields relevant to its semantics. For instance, for wall entities, the geometric term $\mathbf{c}_j = [\mathbf{p}_i^{\text{corner,o},j},\, \mathbf{p}_i^{\text{corner,d},j}]$ captures the relative coordinates of wall endpoints with respect to the ego agent. 

To ensure consistency across entity types, all embeddings use a fixed-length vector equal to the landmark feature dimension. Entity-specific features that do not apply are zero-padded.

Once node embeddings are computed, each agent aggregates relational information from its local agent--entity graph using a GNN. The graph includes bidirectional agent--agent edges and unilateral edges from tasks and obstacles to agents, with pairwise distances serving as edge weights for message passing.

Formally, the relational embedding is obtained as:
\begin{equation*}
x^{(i)}_{\mathrm{agg}} = \mathrm{GNN}\Big(x^{(i)}_i,\, x^{(i)}_j \;\forall j \in \mathcal{N}(i),\, g^{(i)}\Big),
\end{equation*}

where $x^{(i)}_i$ is the ego node embedding, $x^{(i)}_j$ are the neighboring node embeddings within agent $i$’s local graph $g^{(i)}$, and $\mathcal{N}(i)$ denotes the neighborhood of agent $i$.

The combined representation $[o^{(i)},\, x^{(i)}_{\mathrm{agg}}]$ is passed into the policy network, enabling agents to ground local observations in the broader relational context captured by the GNN.

\paragraph{\textbf{Communication}}  
Multi-robot deployments for exploration, monitoring, and search-and-rescue frequently operate under communication limits due to attenuation, interference, and obstructions~\cite{liu2021persistentmonitoring}. As a result, reliable communication is typically restricted to local neighborhoods, and coordination must depend on direct or multi-hop message passing.

To model these realistic constraints, each agent maintains local observations and shares this information only with neighbors within a fixed communication radius. Building on adaptive schemes such as AC2C~\cite{zhang2023ac2c}, we adopt a relay-based multi-hop propagation mechanism that spreads information locally to achieve global awareness while minimizing bandwidth usage. Embedded in our reinforcement learning framework, this mechanism allows agents to balance independent exploration with cooperative updates—learning when information should be shared and which neighbors should receive it—ensuring that fairness and coordination are preserved even under strict communication limits.
%

\paragraph{\textbf{Reward Design}}
Each agent $i$ receives a composite per-step reward that integrates exploration incentives, fairness shaping, workload progress, completion bonuses, and safety penalties.

\vspace{3pt}
\textit{Exploration Reward}
Agents receive a time-decaying bonus when they newly discover a task within their sensing range:
\[
r^{(i)}_{\text{explore}}(t) = 
\begin{cases}
    \eta_0 e^{-\gamma t}, & \text{if a new task is discovered}, \\[4pt]
    0, & \text{otherwise}.
\end{cases}
\]
This encourages agents to discover new tasks early.

\vspace{3pt}
\textit{Fair Assignment Shaping}
The Eisenberg--Gale allocation assigns each agent $i$ a designated task $g_{\pi(i)}$. Agents are penalized by their distance to this task and receive a one-time positive reward upon arrival:
\[
r^{(i)}_{\text{fair}} = -\|p_i - g_{\pi(i)}\| + \mathbf{1}_{\{\text{arrive}\}} R_{\text{arr}},
\]
where $p_i$ is the agent’s position and $R_{\text{arr}}$ is the arrival bonus.

\vspace{3pt}
\textit{Workload Progress}
When an agent arrives and services its assigned task, it receives a reward proportional to the workload completed during that interval. The decrease in workload follows the efficiency-driven dynamics defined in Section~\ref{sec:obs}, where the rate depends on $\text{preference}_{ji}$. Accordingly, the progress reward is expressed as
\[
r^{(i)}_{\text{progress}} = \kappa \cdot \Delta w_{ji},
\quad \Delta w_{ij} = \text{preference}_{ji}\,\Delta t,
\]
where $\kappa$ scales the incremental workload reduction into a per-step reward signal.

\vspace{3pt}
\textit{Completion Bonus}
Upon fully serving the workload at task $j$, the agent receives a completion reward:
\[
r^{(i)}_{\text{complete}} = \mathbf{1}_{\{w_j=0\}} R_{\text{comp}},
\]
where $w_j$ is the remaining workload and $R_{\text{comp}}$ is the completion bonus.

\vspace{3pt}
\textit{Collision Penalties}
Collisions with agents, obstacles, or walls incur negative rewards, discouraging unsafe navigation.
\begin{algorithm}[H]
\caption{EG-MARL Training Loop}
\begin{algorithmic}[1]
\Require Graph observation $G_t$, local observation $o_t$ 

\For{each episode}
  \State Form costs $c_{ij} = \|p_i - g_j\|$ and load prefs/weights
  \State Solve EG to obtain assignment matrix $x$
  \State Set each agent’s target $g_i = \arg\max_j x_{ij}$ \;

    \For{each timestep $t$}
    \State Update $G_t$ and construct $o_t$ via  agent–goal distances, visible-goal features, and sparse communication
    \State Execute action $A_t = \pi_\theta(G_t, o_t)$; step dynamics/workloads
    \State Receive reward 
    \[
        R(s_t, A_t) = \sum_{i=1}^{N} r_t^{(i)} .
    \]
  \EndFor

\EndFor
\end{algorithmic}
\end{algorithm}

During execution, agents operate solely using their learned decentralized policies and local observations; the Eisenberg–Gale program is never solved at runtime.
\section{Stochastic Online Assignment}
In this setting, we assume that agents can communicate without a limited range. The online algorithm operates in two alternating phases—cooperative exploration and assignment. During exploration, agents jointly survey the environment. Once a sufficient number of tasks have been discovered, a partial assignment is solved based on the Eisenberg--Gale program. Agents assigned to specific tasks proceed to service them, while the remaining agents continue exploring. This process repeats until all tasks have been successfully completed.

\textit{Exploration:} 
To ensure efficient exploration, a square grid is superimposed on the map, with grid width set to half the smallest sensing radius among all agents. A global map, shared by all agents, maintains information about which lattice points remain unexplored. Each agent queries this map, computes its distance $d_i$ to all unexplored lattice points, and samples one according to the normalized probability.
\[
P(p_i) = \frac{e^{-d_i}}{\sum_{k \in \mathcal{M}_{\text{unexplored}}} e^{-d_k}},
\]
This formulation ensures that closer lattice points are more likely to be selected. Once an agent selects a lattice point to explore, that point is marked as explored to prevent redundant exploration. As the agents move toward their selected locations, any other lattice points within their sensing radii are also marked as explored. This strategy leads to a cooperative and exhaustive search of the environment. Once a fixed number $k$ (a tunable hyperparameter) of tasks is discovered, the algorithm transitions to the partial assignment phase. If the number of undiscovered tasks is less than $k$, the assignment phase is triggered as soon as all undiscovered targets are found.

\textit{Assignment:}
During the partial assignment phase, the algorithm operates over the currently discovered but unassigned tasks and the currently free agents. 
Let \(\widehat{\mathcal{T}} = \mathcal{T}_{\text{disc}} \setminus \mathcal{T}_{\text{assign}}\) denote the active task set and \(\mathcal{A}_{\text{free}}\) the set of unassigned agents, with \(t = |\mathcal{A}_{\text{free}}|\). 
When exactly \(k\) tasks have been discovered (\(|\widehat{\mathcal{T}}| = k\)), the algorithm enumerates all \(\binom{t}{k}\) possible subsets \(\mathcal{A}_s \subseteq \mathcal{A}_{\text{free}}\) of size \(k\). 
For each subset \(\mathcal{A}_s\), it solves the Eisenberg--Gale optimization problem defined in Eqs.~(1a)--(1c), restricted to agents in \(\mathcal{A}_s\) and the discovered tasks \(\widehat{\mathcal{T}}\). 
The subset achieving the highest equilibrium objective value,
\(\mathcal{A}^\star = \arg\max_{|\mathcal{A}_s| = k} V_{\mathrm{EG}}(\mathcal{A}_s)\),
is selected, and the corresponding agents are assigned according to the optimal EG solution. 
When the number of remaining unassigned tasks satisfies \(|\widehat{\mathcal{T}}| < k\), the algorithm directly solves the same problem over the remaining tasks and currently free agents. 

 \begin{algorithm}[H]
\caption{Stochastic Online Assignment Algorithm}
\label{alg:stochastic-online}
\begin{algorithmic}[1]  
\small
\State \textbf{Input:} Agent set $\mathcal{A}$, task set $\mathcal{T}$, discovery threshold $k$
\State Initialize global exploration map $\mathcal{M}$; set $\mathcal{T}_{\text{disc}}\!\leftarrow\!\emptyset$, $\mathcal{T}_{\text{assign}}\!\leftarrow\!\emptyset$, $\mathcal{A}_{\text{free}}\!\leftarrow\!\mathcal{A}$
\While{$\mathcal{T}_{\text{assign}} \neq \mathcal{T}$}
    \Statex \textbf{Phase I: Cooperative Exploration}
    \For{each agent $i \in \mathcal{A}_{\text{free}}$}
        \State Compute sampling probability:
        \[
        P(p_i) \leftarrow \frac{e^{-d_i}}{\sum_{k \in \mathcal{M}_{\text{unexplored}}} e^{-d_k}}
        \]
        \State Sample unexplored lattice point $p_i \!\sim\! P(p_i)$
        \State Move toward $p_i$ and mark nearby cells in $\mathcal{M}$
        \If{a new task $j$ is discovered}
            \State $\mathcal{T}_{\text{disc}} \!\leftarrow\! \mathcal{T}_{\text{disc}} \cup \{j\}$
        \EndIf
    \EndFor
    \If{$|\mathcal{T}_{\text{disc}} \setminus \mathcal{T}_{\text{assign}}| \ge k$ \textbf{ or } all tasks discovered}
        \Statex \textbf{Phase II: Subset-Based EG Assignment}
        \If{$|\widehat{\mathcal{T}}| = k$}
            \State Enumerate all $k$-agent subsets $\mathcal{A}_s \subseteq \mathcal{A}_{\text{free}}$ of size $k$
            \For{each subset $\mathcal{A}_s$}
                \State Solve $\mathrm{EG}(\mathcal{A}_s,\,\widehat{\mathcal{T}})$ (Eqs.~1a--1c)
            \EndFor
            \State Select subset $\mathcal{A}^* \leftarrow \arg\max_{|\mathcal{A}_s|=k} U(\mathcal{A}_s)$
            \State Assign agents in $\mathcal{A}^*$ to tasks in $\widehat{\mathcal{T}}$ according to the optimal EG solution
        \ElsIf{$|\widehat{\mathcal{T}}| < k$}
            \State Solve $\mathrm{EG}(\mathcal{A}_{\text{free}},\,\widehat{\mathcal{T}})$ (Eqs.~1a--1c) directly and assign
        \EndIf
        \State Update $\mathcal{A}_{\text{free}}$ and $\mathcal{T}_{\text{assign}}$ accordingly
    \EndIf
\EndWhile
\end{algorithmic}
\end{algorithm}

The primary role of the stochastic online EG mechanism is to serve as a strong online baseline with centralized communication against which to evaluate EG-MARL, rather than as a directly deployable algorithm at very large scales. The subset enumeration step is combinatorial in the team size, but the resulting policies are highly informative: they reveal how close decentralized learned policies come to an online EG-guided allocation with full centralized coordination. In larger systems, the same framework could be combined with heuristic or sampling-based subset selection strategies (e.g., random sampling of candidate subsets or prioritizing subsets based on recent travel costs), which we leave for future work.

\section{Experiments}

We empirically evaluate the proposed EG-MARL and Stochastic Online Assignment frameworks under partial observability across varying team sizes (\(N = 3, 7, 10, 15\)). Each configuration uses a proportional map size of \(2.5 \times 2.5\), \(2.7 \times 2.7\), \(2.9 \times 2.9\), and \(3.2 \times 3.2\) respectively, to maintain consistent agent density. The evaluation encompasses three assignment formulations—Eisenberg--Gale, Hungarian, and Min--Max Distance. For the Eisenberg--Gale formulation, we compare decentralized (EG-MARL), online (Stochastic-Online), and centralized implementations, while the Hungarian and Min--Max Distance methods are evaluated only in their centralized forms as baselines.

\subsection{Experimental Setup}
We evaluate our methods in a spatial task-allocation environment adapted from the GraphMPE ``congestion-with-walls'' scenario, built atop the Fair-MARL framework~\cite{aloor2024cooperation}. Each episode initializes $N$ heterogeneous agents, $N$ heterogeneous tasks, static obstacles, and randomly oriented walls within a square environment. The positions of agents, tasks, and obstacles, as well as wall orientations, are resampled every episode to ensure diversity.

\textbf{Comparison between MARL methods.}
In our environment, an agent’s observable neighborhood changes continually as nearby agents, tasks, and obstacles enter or leave its sensing range. To highlight the role of relational representations, we compare RMAPPO (MPE) with Graph\_RMAPPO (Graph\_MPE). RMAPPO encodes visible entities via simple concatenation into a fixed-length vector, a representation that is neither permutation-invariant nor robust to changes in neighborhood size. As shown in InforMARL~\cite{nayak2022scalable} and in our reproduced results (Fig.~2), this causes unstable training and early plateaus because the policy receives syntactically different observations for semantically equivalent configurations. In contrast, Graph\_RMAPPO constructs a local relational graph and performs permutation-invariant message passing over neighbors, enabling consistent representation of dynamically varying entity sets.

\begin{figure}[H]
    \centering
    \includegraphics[width=1.0\linewidth]{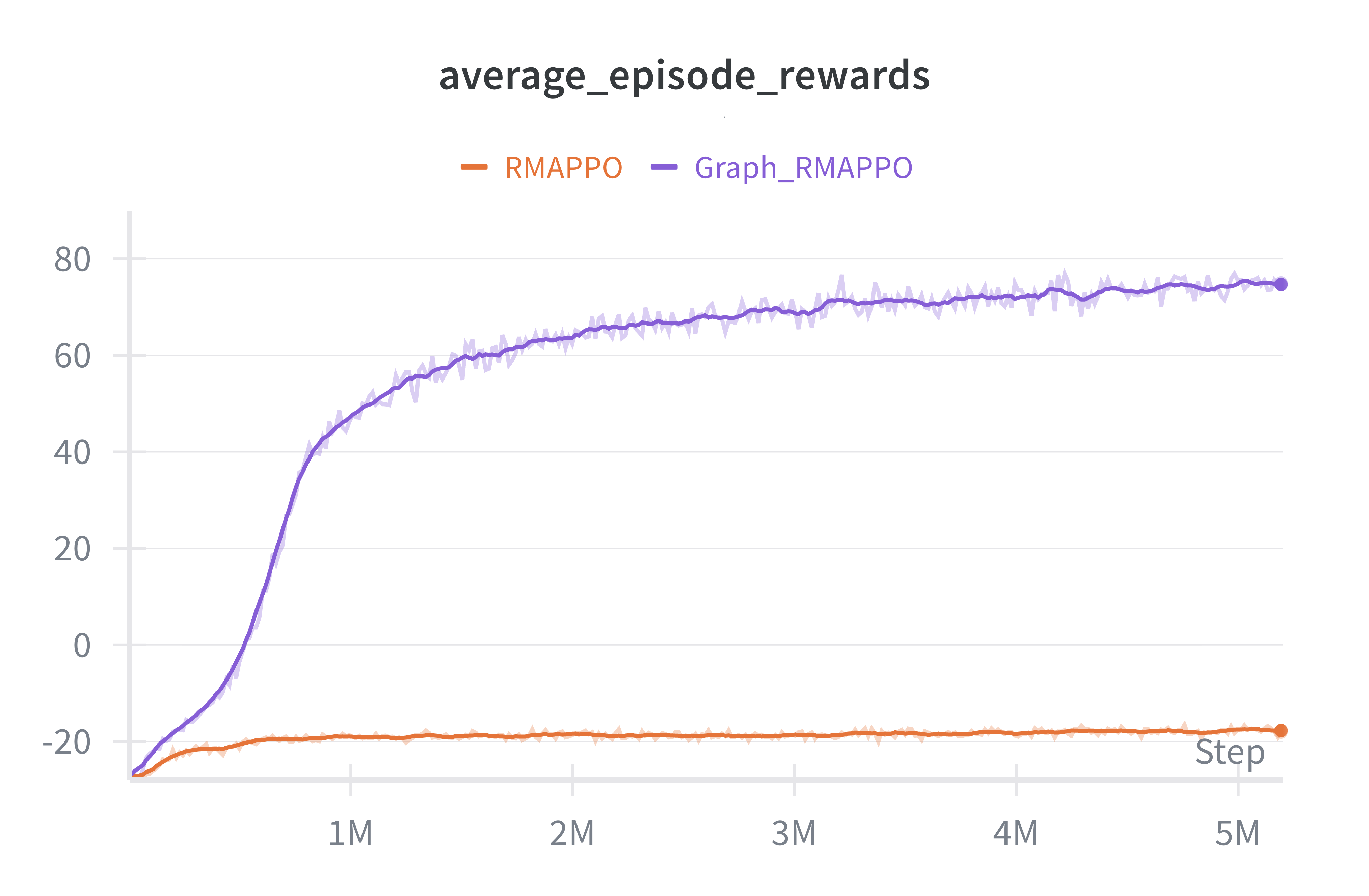}
    \caption{Running-average episode rewards for 7 agents under RMAPPO and Graph\_RMAPPO}
    \label{fig:graph_obs}
\end{figure}
\FloatBarrier

\label{sec:obs}

\subsection{Evaluation Metrics}

We evaluate the efficacy of our algorithms using the following metrics, which collectively capture both efficiency and fairness under decentralized execution.

\begin{enumerate}[leftmargin=*,topsep=2pt,itemsep=2pt,parsep=0pt]
    \item \textbf{Regret.} We measure the regret of any decentralized partially observable policy $\pi$ by comparing the centralized EG objective~\eqref{eq:eg-objective}. Let $d^*_{ij}$ denote the shortest path distance used in the centralized EG program between agent $i$ and task $j$. Accordingly, $u_{ji}(d^*_{ij})$ captures utilities under the shortest centralized distances $d^*_{ij}$. $w_j$ is the task's importance weight.
    The optimal centralized allocation is denoted by $x^* = (x^*_{ji})$, obtained from solving the centralized assignment program.  
    The centralized objective is defined as:
    \[
    U^*_N = \max_{x \in \mathcal{X}} \sum_{j} w_j\, \log\!\Big(\sum_{i} u_{ji}(d^*_{ij})\, x_{ji}\Big),
    \]

For a decentralized, partially observable policy $\pi$, let $d_{ij}$ denote the realized distance traveled by agent $i$ to task $j$. Utility $u_{ji}(d)$ quantifies the benefit obtained by agent $i$ when serving the task $j$ at distance $d$, typically defined as $u_{ji}(d) = (\alpha^{d})\,\text{preference}_{ji}$.
Let $x_{ji}(\pi)$ denote the (possibly stochastic) allocation realized under policy $\pi$. The decentralized utility is then given by:
\[
U_N(\pi) = \sum_{j} w_j\, \log\!\Big(\sum_{i} u_{ji}(d_{ij})\, x_{ji}(\pi)\Big),
\]

The regret of the decentralized partially observable policy $\pi$ is then defined as
\[
R_N(\pi) = \mathbb{E}[U^*_N - U_N(\pi)].
\]
Here, $N$ denotes the number of agents in the system, and the expectation $\mathbb{E}[\cdot]$ is computed over 100 randomized evaluation episodes.

    \item \textbf{Completion Time.} The total number of timesteps required for all workloads to be completed. This metric captures both navigation and service efficiency.

    \item \textbf{Total Distance Traveled.} The cumulative distance traversed by all agents over an episode, reflecting overall motion cost and coordination efficiency.
    
    \item \textbf{Fairness} 
    For each task \( j \), we define the \emph{normalized utility-to-budget ratio} as 
    \( \rho_j = u_j^{\text{realized}} / w_j \), 
    where \( u_j^{\text{realized}} = \sum_i u_{ji}(d_{ij})\,x_{ji}(\pi) \)
    is the total utility achieved for task \( j \) under policy \( \pi \), and \( w_j \) denotes its budget or importance weight.
    \paragraph{Coefficient of Variation (CV)}
    Fairness across tasks is then measured using the reciprocal of the \emph{CV} of \(\{\rho_j\}_{j=1}^{m}\), 
    defined as \( F(\rho) = 1/\text{CV}(\rho) = \mu_{\rho_j}/\sigma_{\rho_j} \),
    where \( \mu_{\rho_j} \) and \( \sigma_{\rho_j} \) denote the mean and standard deviation of the set \( \{\rho_j\} \), respectively. 
    A higher \( F(\rho) \) value indicates more uniform service quality and greater fairness across tasks. This metric captures \textit{distributional fairness}---the equity of realized task outcomes.
    \paragraph{Jain’s Fairness Index}
    We also report Jain’s Index, defined as 
    \( J(\rho) = \frac{(\sum_{j} \rho_j)^2}{m \sum_{j} \rho_j^2} \), 
    which equals 1 under perfectly equalized task outcomes and decreases as disparities grow. 
 
\end{enumerate}

\subsection{Regret under Partial Observability}

We evaluate how closely decentralized policies approximate the centralized equilibrium using regret \(R_N(\pi)\), defined as the expected utility gap between decentralized and centralized objectives. Table~\ref{tab:regret_results} summarizes results across assignment formulations and algorithms for different team sizes.

\begin{table}[H]
\centering
\caption{Regret $R = \mathbb{E}[U^* - U(\pi)]$ for EG-MARL and Online Assignment under the Eisenberg--Gale formulation. Lower values indicate closer alignment with the centralized optimum.}
\label{tab:regret_results}
\resizebox{\linewidth}{!}{
\begin{tabular}{p{1.2cm}|l|c|c|c|c}
\toprule
Assign. & Algorithm & 3 Agents & 7 Agents & 10 Agents & 15 Agents \\
\midrule
\multirow{2}{*}{EG} 
    & MARL & \textbf{0.08} & 3.99 & 6.2 & 9.88 \\
    & Online & 2.22 &  \textbf{0.97} & \textbf{3.7} & 7.05 \\
\midrule
\end{tabular}
}

\end{table}
\FloatBarrier

Following Table~\ref{tab:regret_results}, we observe that EG-MARL achieves the lowest regret for smaller teams ($N=3$). As the number of agents increases ($N=7,10,15$), regret rises for both frameworks, reflecting the greater difficulty of coordinating under partial observability. The Online Assignment performs competitively at $N=7$, $N=10$, and $N=15$.

\subsection{Fairness and Efficiency Analysis}
Table~\ref{tab:eg_hungarian_results} and Figure~\ref{fig:fairness_scatter} summarize results across all configurations. 
For the Eisenberg--Gale (EG) formulation, we compare EG-MARL, Online Assignment, and the Centralized Oracle. 
The Hungarian and Min--Max Distance methods serve as centralized baselines. 
Metrics include completion time ($T$), total distance ($D$), and fairness ($F(\rho)$), where lower $T$ and $D$ and higher $F(\rho)$ indicate better performance.

\begin{table}[H]
\centering
\caption{Comparison across agent counts ($N=3,7,10$) for all assignment rules and metrics include completion time $T$ and total distance $D$ (lower is better), fairness $F(\rho)$, and Jain’s Index $J(\rho)$ (higher is better).}

\label{tab:eg_hungarian_results}

\resizebox{\columnwidth}{!}{%
\begin{tabular}{p{0.7cm}|c|l|c|c|c|c}
\toprule
\# Agents & Assign. & Algorithm & $T$ ($\downarrow$) & $D$ ($\downarrow$) & $F(\rho)$ ($\uparrow$) & $J(\rho)$ ($\uparrow$) \\
\midrule

\multirow{3}{*}{3} 
    & \multirow{3}{*}{EG} 
        & MARL & 2.86 & 4.42 & 3.16 & \textbf{0.83} \\
    &   & Online (k=2) & 6.70 & 11.09 & 3.26 & 0.82 \\
    &   & Centralized & \textbf{2.61} & \textbf{4.03} & 3.14 & \textbf{0.83} \\
\cmidrule{2-7}
    & Hungarian & Centralized & 3.49 & 4.03 & \textbf{3.39} & 0.81 \\
\cmidrule{2-7}
    & Min--Max & Centralized & 4.13 & 4.14 & 2.17 & 0.77 \\
\midrule

\multirow{3}{*}{7} 
    & \multirow{3}{*}{EG} 
        & MARL & \textbf{2.75} & 10.50 & 1.906 & 0.755 \\
    &   & Online (k=3) & 6.53 & 21.74 & \textbf{1.92} & \textbf{0.758} \\
    &   & Centralized & 3.03 & \textbf{10.01} & 1.91 & \textbf{0.758} \\
\cmidrule{2-7}
    & Hungarian & Centralized & 6.65 & 10.13 & 1.45 & 0.64 \\
\cmidrule{2-7}
    & Min--Max & Centralized & 7.07 & 9.92 & 1.51 & 0.66 \\
\midrule

\multirow{3}{*}{10} 
    & \multirow{3}{*}{EG} 
        & MARL & \textbf{2.73} & 16.50 & 1.76 & 0.734 \\
    &   & Online (k=4) & 6.79 & 28.61 & \textbf{1.84} & 0.748 \\
    &   & Centralized & 2.76 & \textbf{15.40} & 1.83 & \textbf{0.749} \\
\cmidrule{2-7}
    & Hungarian & Centralized & 5.51 & 15.48 & 1.35 & 0.622 \\
\cmidrule{2-7}
    & Min--Max & Centralized & 5.84 & 15.56 & 1.33 & 0.614 \\
\midrule

\multirow{3}{*}{15} 
    & \multirow{3}{*}{EG} 
        & MARL & \textbf{3.01} & 26.19 & 1.64 & 0.717 \\
    &   & Online (k=6) & 6.90 & 46.82 & 1.63 & 0.715 \\
    &   & Centralized & 3.10 & \textbf{25.25} & \textbf{1.68} & \textbf{0.727} \\
\cmidrule{2-7}
    & Hungarian & Centralized & 6.14 & 24.9 & 1.195 & 0.572 \\
\cmidrule{2-7}
    & Min--Max & Centralized & 6.77 & 25.4 & 1.196 & 0.568 \\
\bottomrule

\end{tabular}%
} 

\end{table}

\begin{figure}[H]
    \centering
    \includegraphics[width=\linewidth]{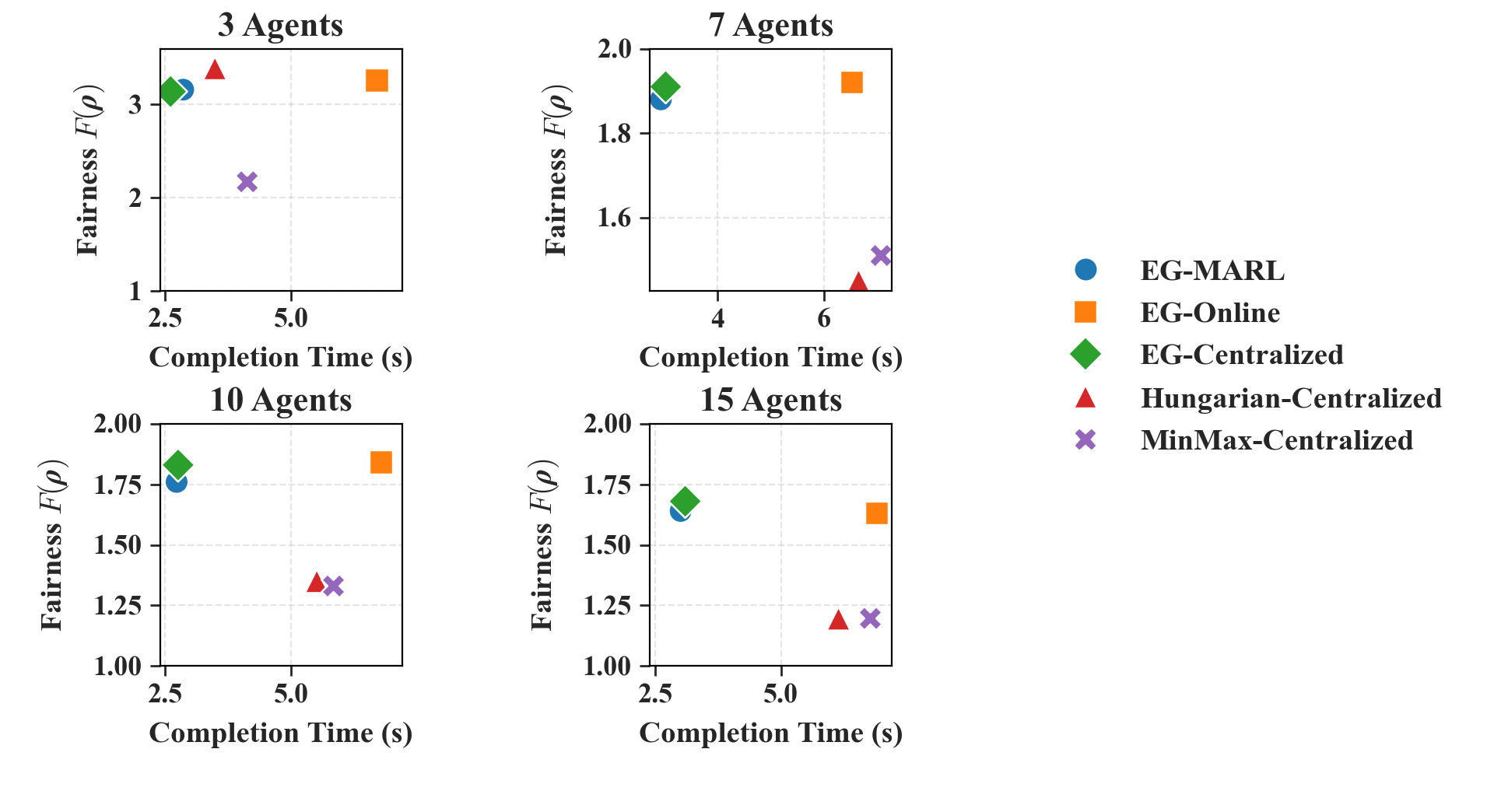}
    \caption{Fairness--Efficiency tradeoff across all algorithms for 3, 7, 10, and 15 agents.}
    \label{fig:fairness_scatter}
\end{figure}
\FloatBarrier

Across all team sizes, both fairness measures—$1/\mathrm{CV}(\rho)$ and Jain’s Index—consistently indicate that the EG-based algorithms deliver more equitable allocations, with one notable exception. In the 3-agent case, fairness exhibits greater variance due to small-sample sensitivity, and the Hungarian baseline achieves slightly higher fairness. We note that the Hungarian does not take into account spatial distance in utility. For the 7, 10 , and 15-agent settings, both fairness metrics strongly demonstrate that the Stochastic-Online and centralized EG formulations attain the highest fairness. The alignment between these two fairness measures provides robust evidence that equilibrium-guided approaches yield more balanced task outcomes than traditional baselines in larger teams. As shown in Fig.~\ref{fig:fairness_scatter}, EG-based methods cluster in the preferred upper-left region—combining higher fairness with lower completion time—while the Hungarian and Min–Max baselines shift lower or rightward in accordance with their efficiency- and distance-driven objectives.

Although Table~\ref{tab:regret_results} shows that EG-MARL's regret increases with team size, Table~\ref{tab:eg_hungarian_results} reveals that it still achieves the fastest completion times for $N=7$ and $N=10$. This outcome suggests that the decentralized EG-MARL policy learns to prioritize preference alignment, completing tasks more quickly, even when this occasionally leads to slightly longer travel distances compared to the centralized solution. The tunable parameter $\alpha$ in the distance-discounted utility $u_{ji} = (\alpha^{d_{ij}})\,\text{preference}_{ji}$ controls the balance between spatial efficiency and preference weighting. In our experiments, we set $\alpha = 0.97$, favoring preference alignment and contributing to faster task completion.

For the 7, 10, and 15-agent cases, the EG optimization reveals complementary strengths across the three methods: \textbf{EG-MARL} achieves the fastest completion times, the \textbf{Stochastic-Online} approach maintains the highest fairness, and the \textbf{Centralized} oracle yields the lowest travel distances. These strengths highlight the trade-offs among scalability, fairness, and efficiency that emerge under different coordination paradigms.

The higher fairness of the Eisenberg--Gale (EG) formulation arises from its logarithmic (concave) utility structure, which, as shown by Amador and Zivan~\cite{cooperative_task_alloc}, naturally incentivizes cooperation among agents. The concavity of the log term reduces marginal gains for already advantaged tasks and encourages balanced allocations.

Table~\ref{tab:eg_k7} analyzes the sensitivity of the Stochastic Online Assignment to the subset size $k$. 
\begin{table}[H]
\centering
\caption{Performance of EG Online for different $k$ values (7 agents)}
\label{tab:eg_k7}
\begin{tabular}{c|c|c|c|c|c}
\toprule
$k$ & Regret ($\downarrow$) & $T$ ($\downarrow$) & $D$ ($\downarrow$) & $F(\rho)$ ($\uparrow$) & $J(\rho)$ ($\uparrow$) \\
\midrule
2 & 2.48 & 7.44 & \textbf{21.16} & 1.91 & 0.755 \\
3 & \textbf{0.97} & 6.53 & 21.74 & 1.92 & 0.758 \\
4 & 1.03 & 6.40 & 23.07 & \textbf{1.95} & \textbf{0.762} \\
    5 & 0.99 & 6.63 & 22.77 & 1.94 & \textbf{0.762} \\
6 & 1.03 & 6.05 & 23.15 & 1.94 & 0.760 \\
7 & 1.41 & \textbf{5.53} & 27.24 & 1.95 & \textbf{0.762} \\
\bottomrule
\end{tabular}
\end{table}
\FloatBarrier
\FloatBarrier

Regret is lowest for moderate subsets ($k=3$–$5$), which balance exploration time with the assignment optimality gap from the centralized solution. As $k$ increases, completion time $T$ decreases monotonically since larger subsets more closely approximate the centralized assignment; when $k=7$, the allocation becomes identical to the centralized EG solution, minimizing idle time but increasing travel distance. In both fairness measures, $F(\rho)$ and Jain’s Index $J(\rho)$, remain largely stable across all subset sizes~$k$, indicating that the online assignment preserves equilibrium balance regardless of how many agents are jointly considered. Minor variations at higher $k$ reflect small sensitivity but do not alter the fairness trend.

\subsection{Ablation Study}

\paragraph{Simple Spread MPE}
In the ablation study, we remove the EG fairness-shaping reward that originally guided agents toward their equilibrium-assigned tasks. Instead, the ablated reward uses a standard MPE-style distance-shaping term that, for each task $j$, computes the minimum distance from any agent to that task:
Intuitively, this mirrors the MPE 'spread' reward: using the minimum agent–task distance encourages the team to spread out so that at least one agent is close to every task. However, because this shaping ignores task priorities, workloads, and heterogeneous preferences, it promotes spatial coverage rather than fairness‑aware assignment.
\[
r_{\text{dist}} = - \min_{i \in \mathcal{A}} \|p_i - g_j\|.
\]
We also evaluate this baseline without the workload–progress reward. While all other rewards remain the same as EG-MARL.
\paragraph{No Exploration Reward} We construct a second EG-MARL variant that removes the exploration reward. 
\begin{table}[H]
\centering
\caption{Ablation results for the 7-agent setting: EG-MARL vs.\ Simple Spread MPE. vs.\ No Exploration Reward.}
\label{tab:ablation_7_agents}
\resizebox{\columnwidth}{!}{
\begin{tabular}{l|c|c|c|c}
\toprule
Algorithm & $T$ ($\downarrow$) & $D$ ($\downarrow$) & $F(\rho)$ ($\uparrow$) & $J(\rho)$ ($\uparrow$) \\
\midrule
EG-MARL & 2.75 & 10.50 & \textbf{1.91} & \textbf{0.755} \\
Simple Spread MPE & 3.62 & 12.08 & 1.895 & 0.737 \\
No Exploration Reward& \textbf{2.72} & \textbf{10.37} & 1.83 & 0.737 \\
\bottomrule
\end{tabular}
}
\end{table} 
We observe from Table~\ref{tab:ablation_7_agents} that removing the EG fairness-shaping increases both completion time and travel distance. Despite relying solely on distance-based shaping, the Simple Spread MPE Baseline maintains reasonable fairness, though it consistently underperforms the EG-MARL model. 

While the No Exploration Reward variant yielded completion time and total distance remaining comparable to those of the full EG-MARL model, its fairness metrics degraded slightly. We hypothesize that, without an exploration incentive, agents prematurely commit to early-discovered tasks, which in turn reduces fairness.

\subsection {Webots Validation in a Warehouse Environment}

\paragraph*{Experimental Setup}
We use the centralized assignemnt method to demonstrate a use-case of our framework in a Webots warehouse-like scene containing a cabinet, pallets, box and gas container. Five heterogeneous differential-drive robots serve as agents, and a stationary supervisor centrally computes assignments to fice differenst tasks. Each robot and workstation has an associated type that determines skill preference and workload priority. The supervisor collects robot positions, workstation locations, and obstacles, and constructs utilities using distance-discounted preference scores. 

The Webots supervisor first constructs a two-dimensional occupancy grid of the warehouse and computes the shortest collision-free path from each robot $i$ to each workstation $j$ using a grid-based A* planner. These path lengths define the distance-aware utilities used in the Eisenberg--Gale assignment, which determines the robot--workstation matching. During execution, the supervisor streams the next waypoint of each robot's precomputed path through low-bandwidth emitter messages, while the robots rely on their onboard go-to-goal controllers to follow the path and avoid obstacles.

\paragraph*{Case Study:  Heterogeneous Robots and Tasks}
To model heterogeneous robot capabilities and the varying urgency of warehouse tasks, we define skill--task compatibility as follows:
\begin{itemize}[nosep, leftmargin=*]
    \item \textbf{Summit-XL Steel}: A heavy-duty platform with the strongest capability for pallet and other high-load tasks.
    \item \textbf{youBot}: Well-suited for delicate manipulation tasks such as handling gas canisters due to its dexterous arm.
    \item \textbf{TurtleBot3 Burger}: A lightweight, agile platform designed for navigating narrow spaces and interacting with drawers, cabinets, and enclosed storage.
    \item \textbf{Pioneer 3-AT}: A robust all-terrain robot with sreliable performance across moderately heavy warehouse tasks.
    \item \textbf{Pioneer 3-DX}: Optimized for fast navigation and high-throughput roles, best with conveyor and open-flow stations.
\end{itemize}
Task-weight (budget) values further encode urgency—high-throughput or safety-critical stations, such as the conveyor or gas canister storage, receive larger weights, while low‑impact stations, such as cardboard boxes, receive lower ones.

Figure 4 illustrates the resulting assignment in our Webots deployment. The Summit-XL Steel is assigned to the wooden pallet stack, the TurtleBot3 Burger to the gas canister, the Pioneer 3-DX to the conveyor platform, the Pioneer 3-AT to the cabinet, and the youBot to the cardboard box.
\begin{figure}[!ht]
    \centering
    
     \caption{Eisenberg--Gale assignment in Webots. The Summit-XL Steel is assigned to the pallet stack, the TurtleBot3 Burger to the cabinet, the Pioneer 3-DX to the conveyor, the youBot to the gas canister, and the Pioneer 3-AT to the cardboard box.}
     \vspace{6pt}
    \begin{subfigure}[t]{\columnwidth}
        \centering
        \includegraphics[width=\columnwidth]{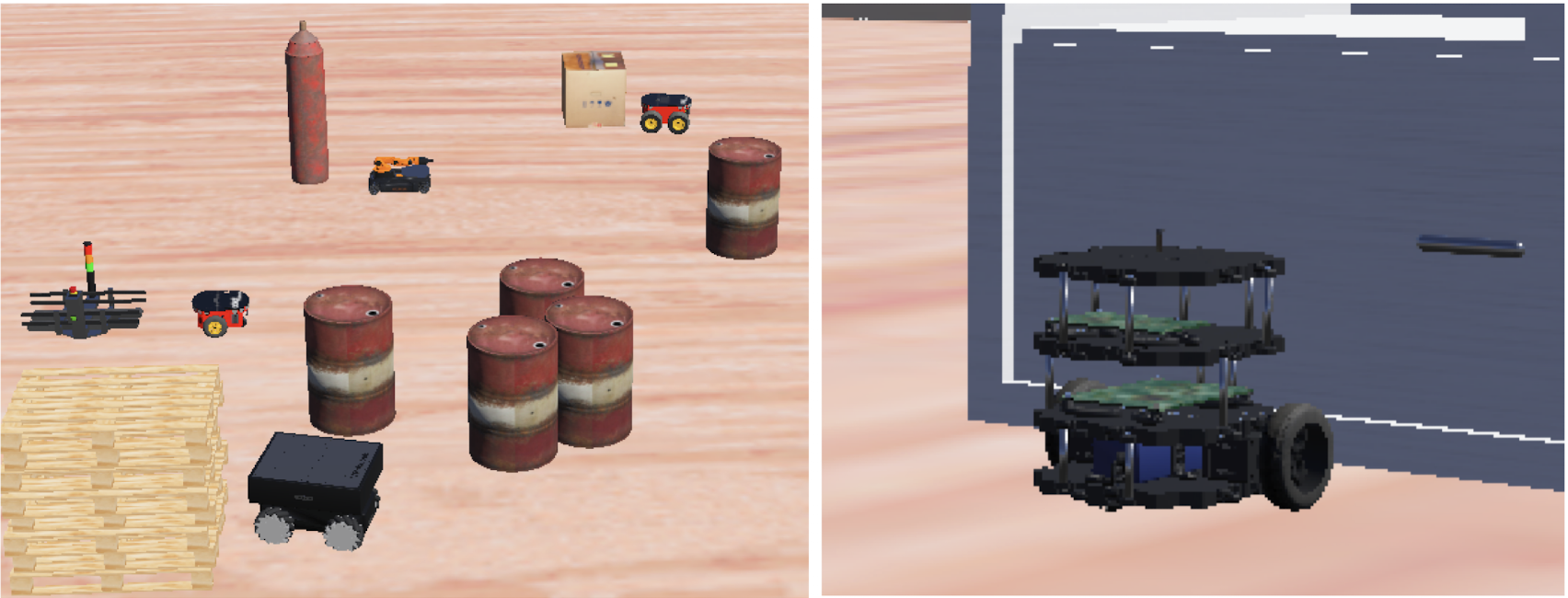}
    \end{subfigure}
   
    \label{fig:appendix_webots}
\end{figure}

\FloatBarrier

\paragraph*{Performance Evaluation}
We evaluate the Webots warehouse deployment using the same fairness and efficiency metrics as in our MPE experiments.
\begin{table}[!ht]
\centering
\caption{Webots Evaluation Metrics for EG, Hungarian, and Min–Max. Bold values indicate the best performance per column.}
\label{tab:webots_eval}
\begin{tabular}{l|c|c|c|c}
\toprule
Method & $T$ ($\downarrow$) & $D$ ($\downarrow$) & $F(\rho)$ ($\uparrow$) & $J(\rho)$ ($\uparrow$) \\
\midrule
EG         & 54.53 & 36.82 & \textbf{1.710} & \textbf{0.745} \\
Hungarian  & \textbf{51.78}          & 34.27 & 1.259          & 0.613 \\
Min--Max   & 70.85          & \textbf{36.86} & 1.438 & 0.674 \\
\bottomrule
\end{tabular}
\end{table}
\FloatBarrier

\paragraph*{Summary of Findings}
In this case study, the centralized Eisenberg--Gale (EG) allocation achieves the strongest overall balance between fairness and efficiency. EG attains the highest fairness scores ($F(\rho)$ and $J(\rho)$) while also producing the competitive completion time among the three assignment methods. The Hungarian method, despite achieving the best completion time and competitive travel distance, exhibits substantially lower fairness due to prioritizing preference alignment without regard to the priority weight across tasks. The Min--Max baseline yields the shortest total distance traveled but shows lower fairness compared to EG, as minimizing worst-case travel does not incorporate task urgency or skill--task compatibility. These results are consistent with our simulation findings and demonstrate that EG equilibrium-guided assignment provides more equitable task outcomes while maintaining competitive spatial efficiency in realistic heterogeneous multi-robot deployments. This warehouse scenario mirrors our navigation-and-service formulation: tasks correspond to shelf locations with item requests, and heterogeneous robots must navigate around obstacles to service them. Although this experiment uses centralized supervision rather than learned decentralized policies, it demonstrates that the fairness–efficiency trade-offs observed in MPE simulations carry over to a realistic, heterogeneous multi-robot warehouse setting.

\section{Assumptions and Limitations}

Our formulation makes several modeling assumptions. We consider one-to-one assignments with no preemption, and tasks are static in our experiments. EG-MARL assumes local, range-limited communication with multi-hop propagation, while the Online Assignment method operates under full centralized communication and serves as a baseline. We remark that the computational bottleneck of the Stochastic Online Assignment is dominated by the combinatorial enumeration term $\binom{N}{k}$, resulting in a complexity of $T = \mathcal{O} \left( \binom{N}{k} \cdot \text{Poly}(k) \right)$. This exponential dependence on the team size $N$ limits its practical use to systems with a small number of agents. In addition, we recognize the statistical bias from the testing samples. Finally, we acknowledge that the test environments for different team sizes were not held constant, which may introduce sample bias in the results across settings.

Finally, we clarify that EG-MARL is not an approximation algorithm for the underlying Dec-POMDP, but rather an empirically EG-guided learning framework that uses the equilibrium solution only as a supervisory signal during training.

\section{Conclusion and Future Work}

This work presented two complementary frameworks that bridge competitive equilibrium theory with decentralized coordination in spatial multi-agent systems. The proposed EG-MARL framework integrates the Eisenberg--Gale equilibrium into a CTDE training paradigm, enabling agents to achieve near-centralized efficiency and fairness under partial observability through equilibrium-guided reward shaping. In parallel, the Stochastic Online Assignment algorithm extends equilibrium-based reasoning to dynamic exploration settings, allowing agents to perform subset-based assignments as tasks are discovered in real time. Together, these methods demonstrate that decentralized agents can approximate Pareto-efficient and fairness-aware outcomes without relying on global state information.

Future work will explore extending this framework to settings with continuously arriving tasks and evolving workloads, where both urgency and availability vary over time. Another direction involves incorporating adaptive communication learning into EG-MARL, allowing agents to selectively share task and state information under bandwidth constraints. Finally, extending the Eisenberg--Gale formulation to cooperative servicing—where multiple agents jointly complete a single task—offers a natural avenue for modeling fractional workloads and dynamic fairness in fully cooperative, heterogeneous systems.
Another important direction is to incorporate agent-side fairness, recognizing that real multi-robot teams operate under heterogeneous resource constraints such as battery life and energy usage. Addressing these disparities may require new equilibrium formulations or multi-objective fairness criteria that jointly balance task with equitable agent-level fairness.

\end{document}